\documentclass[aps,pra,showpacs,twoside,twocolumn,10pt, longbibliography]{revtex4-2}
\usepackage[colorlinks=true, citecolor=red, urlcolor=blue ]{hyperref}
\usepackage{times,epsfig,amssymb,amsfonts,amsmath,bm,subfigure,mathtools,amsthm,braket,soul,enumitem,color,physics,graphics,graphicx}
\usepackage[normalem]{ulem}
\usepackage{comment}
\usepackage{epstopdf}
\newtheorem{theorem}{Theorem}

\newtheorem{proposition}{Proposition}

\newtheorem{lemma}{Lemma}

\newtheorem{observation}{Observation}
\newtheorem{protocol}{Protocol}

\usepackage{optidef}

\begin{document}

\title{
Classical capacities under physical constraints: More capacity with less entanglement
}

\author{Sudipta Mondal$^1$, Pritam Halder$^{1}$, Saptarshi Roy$^2$, Aditi Sen(De)$^1$}

\affiliation{$^1$Harish-Chandra Research Institute, A CI of Homi Bhabha National Institute, Chhatnag Road, Jhunsi, Prayagraj 211019, India\\
\(^2\)  QICI Quantum Information and Computation Initiative, Department of Computer Science,\\
The University of Hong Kong, Pokfulam Road, Hong Kong }

\begin{abstract}

Current advancements in communication equipment demand the investigation of classical information transfer over quantum channels, by encompassing realistic scenarios in finite dimensions.  To address this issue,  we develop a framework for analyzing classical capacities of quantum channels where the set of states used for encoding information is restricted based on various physical properties. Specifically, we provide expressions for the classical capacities of noiseless and noisy quantum channels when the average energy of the encoded ensemble or the energy of each of the constituent states in the ensemble is bounded. In the case of qubit energy-preserving dephasing channels, we demonstrate that  a nonuniform probability distribution based on the energy constraint maximizes capacity, while we derive the compact form of the capacity for equiprobable messages.  We suggest an energy-constrained dense coding (DC) protocol that we prove to be optimal in the two-qubit situation and obtain a closed-form expression for the DC capacity.  Additionally, we demonstrate a no-go result, which states that when the dimension of the sender and the receiver is two, no energy-preserving operation can offer any quantum advantage for energy-constrained entanglement-assisted capacity. We exhibit that, in the energy-constrained situation, classical-quantum noisy channels can show improved capabilities under entanglement assistance, a phenomenon that is unattainable in the unrestricted scenario.

\end{abstract}	

\maketitle

\section{Introduction}
\label{sec:intro}

In the mid-twentieth century, Shannon revolutionized communication technology by linking it with information theory \cite{Shannon1948,Shannon1949,Cover2005}. Interestingly, the last three decades have witnessed a second quantum revolution \cite{Dowling2003,Nielsen2012,Jaeger2018}, advancing quantum technologies such as quantum communication \cite{Bennett1992,teleportation_1993, Ekert91, gisin2002_crypto}, quantum computing capabilities \cite{Grover1996,Shor1997,Bennett2000,one_way_prl2001}, quantum sensing \cite{Giovannetti2006,Degen2017}, surpassing classical methods. In the field of quantum communication, a prominent research direction includes determining the classical capacity of a quantum channel, the highest rate at which classical information may be successfully transferred using quantum resources
\cite{cc1,Schumacher1997, 
Holevo1998,Bennett1999,bennett2002entanglement,cc2, Holevo2020}. Specifically, a shared entangled state, a protocol known as dense coding (DC) \cite{Bennett1992,Bose2000,Hiroshima_2001,Ziman2003,Bruss2004} was demonstrated to bypass the no-go result by Holevo (a bit per qubit)  \cite{Holevo1973}.  More crucially, the implementation of these protocols in laboratories facilitates the theoretical development of quantum communication devices to convey classical information \cite{Mattle1996,Schaetz2004,Banaszek2004,Chiuri2013,Williams2017,Kumar2019}.

The majority of studies on the classical capacity of quantum channels assume ideal conditions such as unrestricted access to encoding states, the ability to utilize the channel multiple times, and access to unlimited shared entanglement (in the case of entanglement-assisted capacities). A few exceptions are the study of single-shot capacities \cite{Wang2012} and capacities with bounded entanglement-assistance \cite{Shor2004,Zhu2017}.

The current theoretical and experimental advances in communication devices necessitate constraints at various stages of information transfer. 
For example, one such physically significant limitation arises from thermodynamics, where the energy of encoded states is finite, limiting the set of states utilized to encode classical signals. Notably, capacities with such constraints have traditionally been considered in continuous variable systems \cite{Braunstein2000,Ban2000,holevo2001evaluating,Ralph2002,aditi2005,Aditi2007,Patra2022,cvqc1}  in order to prevent divergent classical capacities caused by infinite dimensions. 
More crucially, this raises an intriguing question: \textit{what impact do energy constraints have on capacities when dealing with finite-dimensional systems?}
 In general,  restrictions on the accessibility of encoding states might arise from other bounded state properties like purity, coherence, magic, etc., depending on the particulars of the physical situation at hand. The primary goal of this work is to determine the classical capacity of quantum channels under the restrictive conditions of the encoding procedure, hence addressing the questions.

 By introducing the notion of the restricted classical capacity, we reveal that, unlike continuous variable systems, there is an emergent energy scale beyond which the energy bound becomes irrelevant. Interestingly, we exhibit that for a noisy channel, especially a dephasing one, the capacities are functions of non-uniform energy-dependent probabilities of the optimal ensemble, resulting in different classical capacities when the energy restrictions are average or strict, or when messages are equiprobable. Beyond a single system, we present our findings in the energy-constrained entanglement-assisted scenario. In particular, we consider an energy-constrained DC protocol that ensures quantum advantage with the help of pre-shared entanglement. This protocol yields a closed-form expression of the DC capacity, which is optimal for two local dimensions. However, in the energy-constrained entanglement-assisted protocol, we prove that passive operations cannot provide any quantum benefit. The energy-constrained framework enables us to demonstrate that a classical-quantum (CQ) channel offers enhanced DC capacity via entanglement, a feature shown to be absent in the unrestricted (conventional) communication paradigm.   

Our paper is organized as follows. The definition of energy-constrained entanglement-unassisted classical capacity of a quantum channel is presented in Sec.~\ref{sec:formulation}. We provide the expressions of the classical capacity for both the noiseless and the noisy cases in Sec.~\ref {subsec:noise_less} and Sec.~\ref{subsec:noisy_channel}, respectively. The energy-constrained DC protocol is proposed in Sec.~\ref{sec:canonical_DC} while Sec.~\ref{subsec:shirokov} highlights the advantage of pre-shared entanglement on the capacity of a CQ channel. We conclude in Sec.~\ref{sec:conclusion}.

\section{Constrained classical capacities}
\label{sec:formulation}

In a classical communication scheme through the quantum state, the encoding of the classical message $ i$ is performed on the state $\rho_i$, which can have a bounded value of some physical or information-theoretic quantity. In a constrained classical information transmission protocol, two constraints can naturally appear -- \((1)\) \emph{average} and \((2)\) \emph{strict}. In particular, for every physical property $\mathcal{P}_k$, when an encoding ensemble of $\{p_i,\rho_i\}$ is used for transmitting classical information, an \emph{average} constraint signifies $\sum_i p_i \mathcal{P}_k(\rho_i) \leq {\tt P}_k.$
When $\mathcal{P}_k$ is a linear functional like energy, i.e., $\mathcal{P}_k(\rho) := \text{Tr}(\mathcal{O}_k \rho)$ for some observable $\mathcal{O}_k$, the average constraint translates to $ \mathcal{P}_k(\sum_i p_i\rho_i = \bar\rho) = \text{Tr}(\mathcal{O}_k \bar\rho)\leq {\tt P}_k$, where $\bar\rho$ is the average ensemble state.
On the other hand, in the \emph{strict} case, the considered physical property $\mathcal{D}_k$  should satisfy $\mathcal{D}_k(\rho_i) \leq {\tt D}_k, \forall i$.
We are now ready to define the constrained classical capacity of a quantum channel $\Lambda$, where $N_A$ physical properties $\{\mathcal{P}_k\}_{k=1}^{N_A}$ with $\{{\tt P}_k\}_{k=1}^{N_A}$ values are constrained on an average, and $N_S$ physical properties $\{\mathcal{D}_k\}_{k=1}^{N_S}$ having
 $\{{\tt D}_k\}_{k=1}^{N_S}$ are strictly constrained. Therefore, for a total of $N_A+N_S$ constraints, we have
 \begin{equation}
\begin{aligned}
C_{\{{\tt P}_k\}_{\text{A}}\{{\tt D}_k\}_{\text{S}}}(\Lambda):=  & \max_{\{p_i,\rho_i\}} \chi(\{p_i,\Lambda(\rho_i)\})\\
\textrm{subject to}  \quad & \sum_i p_i \mathcal{P}_k(\rho_i) \leq {\tt P}_k, &k = 1, \ldots N_A, \\
& \textrm{and} \quad     \mathcal D_k(\rho_i) \leq {\tt D}_k ~~\forall \rho_i, & k = 1, \ldots N_S. 
\end{aligned}
\label{eq:cc1}
\end{equation}
Here $\chi(\{q_j,\sigma_j\}) = S(\sum_j q_j \sigma_j) - \sum_j q_j S(\sigma_j)$ represents the Holevo capacity \cite{Holevo2020,Schumacher1997,cc1,cc2} of the ensemble $\{q_j,\sigma_j\}$, and $S(\mu) = -$Tr$(\mu \log \mu)$ is the von-Neumann entropy of $\mu$. On several occasions, we exhibit that capacity-maximizing probabilities turn out to be functions of constraints and the quantum channel, i.e., $p_i = f_i(\{{\tt P}_k\},\{{\tt D}_k\},\Lambda)$. Although mathematically optimal, one may argue that this is an unphysical situation since the probabilities of the classical messages are typically independent of the constraints and noisy channel, and hence, in such situations, we additionally focus on a variant of the capacity, with an additional imposition that the encoding probabilities are uniform.


\subsection{Energy-constrained classical capacity of noiseless quantum channel}
\label{subsec:noise_less}

Let us now concentrate on the energy-constrained (both average and strict) classical capacity of noiseless quantum channels. The Hamiltonian for the $d$-level system associated with the sender's system is taken to be 
  \(      H = \sum_{n = 0}^{d-1} n \ketbra{n}{n}\).
 We can prove the exact form of energy-constrained (EC) classical capacity in the following theorem.
  \begin{theorem}
   \label{theorem:eccp1}
       The energy-constrained classical capacity of a noiseless quantum channel in arbitrary dimension for both average and strict constraints is identical and can be expressed as
       \begin{eqnarray}
             C_{{\{E\}}_{\emph{A(S)}}}(\Lambda = \mathbb I) =:  C_{\{E\}}  = \Bigg\{
        \begin{array}{cc}
           \log_2 d   & E \geq \mathbb E_d, \\ 
            H(\{p_n^E\})  & E < \mathbb E_d,
        \end{array}
        \label{eq:eccc3}
       \end{eqnarray}
where $\mathbb E_d = \frac{d-1}{2}$. Here $H(\{p_x\})= -\sum_x p_x \log_2 p_x$ is the Shannon entropy of the probability distribution $\{p_x\}$,  where
\begin{eqnarray}
    p_n^E = \frac{e^{-n \beta_E}}{\sum_{n=0}^{d-1}e^{-n \beta_E}}, ~~~n = 0, \ldots, d-1,
    \label{eq:thermalcoefficients}
\end{eqnarray}
are the weights of the $d$-dimensional thermal state $\tau_E = \sum_{n=0}^{d-1} p_n^E \ketbra{n}{n},$ with average energy $E$.  Here, $\ket n $ denotes the $n^{\emph{th}}$ eigenstate of the Hamiltonian $H$.
 The inverse temperature, $\beta_E=\frac{1}{K_{B}T_{E}}$ (with \(T_{E}\) and \(K_{B}\) being the temperature and Boltzmann constant respectively) is obtained from $\emph{Tr}(\tau_E H) = \sum_{n=0}^{d-1} np_n^E = E$.
 \end{theorem} 
\noindent  See Appendix~\ref{app:1} for detailed proof of this Theorem, where we also demonstrate the dimensional advantage of capacity for $E < \mathbb E_d$ in Appendix~\ref{app:dimadv}. From Theorem~\ref{theorem:eccp1}, an interesting observation emerge -- for any dimension $d$, there is an emergent energy scale $E = \frac{d-1}{2}= \mathbb{E}_d$ above which the capacity $C_{\{E\}}$ becomes insensitive to the energy bound. Since $E \geq\mathbb{E}_d$, the set of available ensembles supports a full orthogonal basis where each element of the ensemble satisfies the energy bound, thereby respecting both the average and strict energy bounds. A non-trivial effect of energy constraint on capacity is only observed when $E < \mathbb E_d$. Note that for the infinite-dimensional case $(d \to \infty)$, $\mathbb{E}_d \to \infty$ and hence this energy scale becomes irrelevant. Furthermore, for any finite $d<\infty$, for any Hamiltonian $\tilde{H}$ that satisfies $0 < ||\widetilde H||< \infty$, the emergent energy scale still persists with $\widetilde{\mathbb E}_d = \frac1d$Tr $\widetilde H$.


\subsection{Energy-constrained classical capacity of noisy quantum channels}
\label{subsec:noisy_channel}

We now focus on the energy-constrained classical capacity of noisy quantum channels. In particular, for our investigation, we restrict our attention to energy-preserving channels in $d=2$.
The most general qubit energy-preserving channel is unital \cite{Chiribella2017, Swati2023}, given by
$ \Lambda^{\tt EP}(\rho) = \sum_i p_i \mathcal{U}_i(\rho)$,
where $\mathcal{U}_k(*) \equiv U_k (*) U_k^\dagger$. The property of energy preservability induces a stronger requirement on $U_k$s. 
\begin{lemma}
A qubit energy-preserving channel is unital and can be written as a mixture of unitaries $\{p_k,U_k\}$, where all $U_k$s are energy preserving, i.e., $[U_k,H]=0$ for all $U_k$s.
\label{lemma:epc}
\end{lemma}
\noindent The details of the proof are provided in Appendix~\ref{app:lemma1}.

The most general form of a qubit energy-preserving unitary for the Hamiltonian is $U = \cos \theta ~\mathbb I +i \sin \theta ~\sigma_z$, where $\sigma_z$ is the Pauli $Z$ operator. 
Here, without loss of any generality, the $z$-direction is taken to be along the eigenbasis of the system Hamiltonian.
This dictates the form of a qubit energy-preserving channel to be 
\begin{eqnarray}
    \Lambda^{\tt EP}(\rho)  = \sum_j q_j \mathcal{U}_j(\rho),
\end{eqnarray}
with $U_k =\cos \theta_k ~\mathbb I +i \sin \theta_k \sigma_z = e^{i\cos \theta_k \sigma_z}$, and $\sum_k q_k = 1$.
Let us briefly discuss a special case of energy-preserving channels, the dephasing channel $\Lambda^{\tt Deph}_\lambda$ whose action on an arbitrary state $\rho$ is given by 
\(\Lambda^{\tt Deph}_\lambda(\rho) = \lambda \rho + (1-\lambda)\sigma_z \rho \sigma_z\). In the unrestricted framework of communication,
classical capacity of $\Lambda^{\tt Deph}_\lambda$ irrespective of the value of $\lambda$, is given by \( C(\Lambda^{\tt Deph}_\lambda) = 1 ~~\forall \lambda \in [0,1]\), 
since $\Lambda^{\tt Deph}_\lambda$ preserves the $\{\ket{0}, \ket{1}\}$-subspace. 

Now, the question is whether the noisy classical capacities change in the presence of an energy constraint. The energy-constrained capacity of $\Lambda^{\tt Deph}_\lambda$ under the average energy constraint for $E \geq \frac12$, remains unity, $C_{\{E\}_{\text{A}}}(\Lambda^{\tt Deph}_\lambda) = 1$ $\forall\lambda$. For $E < \frac12$, we recover the noiseless value of  $C_{\{E\}_{\text{A}}}(\Lambda^{\tt Deph}_{\lambda})= H(E)$. Note that there is a crucial difference between the two. In the noisy case, the optimal ensemble, $$\{p_0 = 1-E, \rho_0 = \ketbra{0}{0}; p_1 = E, \rho_1 = \ketbra{1}{1} \},$$ for achieving this capacity is unique while this is not the case in the noiseless scenario. Quite remarkably, note that a non-uniform probability distribution maximizes the capacity, and the probabilities are dependent on the energy constraint. 
Typically, the classical messages come from independent sources, whereas extracting the maximal capacity requires the message probabilities to correlate perfectly with the energy constraint, thereby failing to capture the physically relevant situation. 

Motivated by the above observation, we introduce an alternative version of capacity, alongside the standard constrained capacity in Eq.~\eqref{eq:cc1}, wherein the classical messages are assumed to be equiprobable.
In this case, for $E \geq \frac12$,  the capacity remains unity like before. However, for
$E < \frac12$, the capacity can be expressed as the Holevo capacity of the ensemble $\{\frac12,  \Lambda^{\tt Deph}_\lambda(\rho_+);\frac12, \Lambda^{\tt Deph}_\lambda(\rho_-)\}$, where the pure states  $\rho_\pm$, used for encoding the equiprobable messages are
\begin{eqnarray}
    \rho_\pm = {\tt Proj}\Big(\sqrt{1-(E \pm \Delta)}\ket{0} \pm \sqrt{E \pm \Delta}\ket{1}\Big).
\end{eqnarray}
The optimal value of $\Delta$ depends on $\lambda$ and $E$. (see Appendix~\ref{app:3}).
Notably, a closed-form expression of the EC capacity for the complete dephasing channel (which corresponds to $\lambda = \frac12$) can be obtained as
  \begin{eqnarray}
             \widetilde{C}_{\{E\}_{\text{A}}}(\Lambda^{\tt Deph}_{\lambda = \frac12}) = \Bigg\{
        \begin{array}{cc}
           1   & E \geq \frac12, \\ 
            H(E) - \frac{H(2E)}{2}  & E < \frac12.
        \end{array}
        \label{eq:ecccdpc1}
       \end{eqnarray}
       Here, and throughout the manuscript, we will use $\widetilde{C}$ to denote capacities in the setting of equiprobable messages. 

When one considers the strict energy constraint, an even stronger separation appears, where for all $E$ $<1$, we get a strictly lower value of the capacity compared to the noiseless case. In particular, for $\lambda = \frac12$, we have
 \begin{eqnarray}
      C_{\{E\}_{\text{S}}}(\Lambda^{\tt Deph}_{\lambda = \frac12}) \geq  \widetilde{C}_{\{E\}_{\text{S}}}(\Lambda^{\tt Deph}_{\lambda = \frac12}) =
            H\Big(\frac E2 \Big) - \frac{H(E)}{2}.
        \label{eq:ecccdpc2}
 \end{eqnarray}
 The equality is achieved when \( E = 1 \) for any value of \( \lambda \), and for all values of \( E \) when \( \lambda = 0 \) or \( \lambda = 1 \) (the noiseless case). In the noisy case $(\lambda \neq 0,1)$, for $E<1$, even in the strict case, the capacity maximizing probabilities are non-uniform, and hence $C_{\{E\}_{\text{S}}}(\Lambda^{\tt Deph}_{\lambda}) >  \widetilde{C}_{\{E\}_{\text{S}}}(\Lambda^{\tt Deph}_{\lambda})$ and the corresponding optimal probabilities are not expressible as neatly as in the case of average energy constraint
(see Appendix~\ref{app:3} for details including plot of $C_{\{E\}_{\text{S}}}(\Lambda^{\tt Deph}_\lambda)$ against $\lambda$ for different values of $E$). We find that a similar analysis can be carried out when two constraints, like average energy and purity, can be involved for obtaining classical capacity (see Appendix~\ref{app:dual}).  Here, two distinct energy scales emerge, splitting the available energy range into three markedly different domains characterized by qualitatively different behaviors of capacity.

It is important to note here that in the unrestricted scenario, a perfect qubit channel's communication utility (in terms of classical capacity) was identical with all channels $\Lambda^{\tt Deph}_\lambda$. Such an equivalence breaks down in the energy-constrained setting, especially when a strict energy constraint is imposed and the messages are equiprobable.  \\

\section{Energy-constrained Dense coding: \\ more capacity with less entanglement}
\label{sec:canonical_DC}

Pre-shared entanglement has been demonstrated to enhance the communication utility of quantum channels for sending classical information, known as quantum dense coding \cite{Bennett1992, Hiroshima_2001,Ziman2003, Bruss2004}.
In this protocol, the sender encodes the message, $i$, on one part of the shared entangled state $\rho_{SR}$ using a unitary operation $\mathcal{U}_i$, the set of states $(\mathcal U_i \otimes \mathbb I) \rho_{SR}:= \rho_{SR}^i$ with the marginal state at sender's node $\rho_S^i:= $Tr$_R(\rho_{SR}^i)$. The  constrained DC capacity for an arbitrary channel $\Lambda$ can be expressed as
\begin{equation}
\begin{aligned}
C^{\tt DC}_{\{{\tt P}_k\}_{\text{A}}\{{\tt D}_k\}_{\text{S}}}(\Lambda):=  & \max_{\{p_i,\mathcal{U}_i,\rho_{SR}\}} \chi(\{p_i,\Lambda(\rho^i_{SR})\})\\
\textrm{subject to}  \quad & \sum_i p_i \mathcal{P}_k(\rho^i_S) \leq {\tt P}_k, ~~~~k = 1, \ldots N_A, \\
    & \mathcal D_k(\rho^i_S) \leq {\tt D}_k ~~\forall \rho_i, ~~~~~~~k = 1, \ldots N_S. 
\end{aligned}
\label{eq:dccc}
\end{equation}
For the noiseless case with average energy constraint, $C^{\tt DC}_{\{{\tt P}_k\}_{\text{A}}\{{\tt D}_k\}_{\text{S}}} := C^{\tt DC}_{\{{\tt P}_k\}_{\text{A}}\{{\tt D}_k\}_{\text{S}}}(\Lambda = \mathbb I)$. Note that, in principle, we are left with an optimization problem whose parameter space grows as $\mathcal O(d^4).$ 

\begin{proposition}
\label{obs:dce1}
  The original DC strategy without constraints violates the energy bound for $E < \mathbb{E}_d$ and, hence, is not suitable for implementing an energy-constrained DC protocol in this regime. 
\end{proposition}
\begin{proof}
The original DC protocol \cite{Hiroshima_2001} in dimension $d$ involves the sender to encode classical messages of alphabet size $d^2$, 
denoted by $\{k,m\}_{k,m = 0}^{d-1}$, using local unitaries $U_{km} = X^kZ^m$ composed of the shift $X$ and phase operators $Z$ with equal probabilities. The action of these unitaries on the arbitrary energy eigenstate $\ket n$ is given by \( X \ket n = \ket{n+1 ~\rm{mod} ~d}, \quad Z \ket n = \omega^{n} \ket n,\) where $\omega = e^{2\pi i/d}$ is the $d^{th}$ root of unity.
Even when the pre-shared entangled state $\rho_{SR}$ is chosen such that Tr$(\rho_SH)< \mathbb E_d$, the average state transmitted by the sender after encoding is $\bar \rho_S = \frac{1}{d^2}\sum_{k,m=0}^{d-1} U_{km} \rho_S U_{km}^\dagger = \frac{\mathbb I}{d}$, whose energy is Tr$(\frac 1 d H) = \frac{d-1}{2} = \mathbb E_d$, thereby violating the energy bound.
\end{proof}

\begin{protocol}\emph{(Energy-constrained dense coding.)}
Consider a bipartite state, $\ket{\psi_{SR}}$, with local dimension $d$. The sender encodes $d^2$ classical messages $\{k,m\}_{k,m = 0}^{d-1}$ via local unitaries $U_{km} = X^kZ^m$ as used in the unrestricted setting. The energy-constrained protocol differs from the usual one in the choice of the encoding probabilities and the shared state. \\
\emph{Case 1.}
For $E < \mathbb E_d$, the shared state is chosen to be of the form $\ket{\psi_{SR}^{E'}} = \sum_{n=0}^{d-1} \sqrt{p_n^{E'}}\ket{n,n}$, where $p_n^{E'}$ are the weights of a local thermal state with average energy $E'$, see Eq.~\eqref{eq:thermalcoefficients}. From the average energy constraint, $E' \leq E$.
Now, from the Proposition. \ref{obs:dce1},  equiprobable encoding fails to satisfy the energy bound. Hence, we adopt the following strategy: All encoding unitaries $V_k := \{U_{km}\}_{m=0}^{d-1}$ are implemented with a probability of $\frac{q_k}{d}$, where $\sum_{k=0}^{d-1}q_k = 1$. Here, $\{q_k\}$ and $E'$ have to be optimally chosen depending on the value of the energy constraint $E$.  \\
\emph{Case 2.} For $E \geq \mathbb E_d$, the usual DC protocol with equiprobable messages is employed with a maximally entangled shared state, $\ket{\psi_{SR}} = \frac{1}{\sqrt{d}}\sum_{n=0}^{d-1}\ket{n,n}$ to obtain the maximum capacity.
\label{protocol:ecdc}
\end{protocol}
Some features of this protocol can immediately be noted: First, the particular choice of $\ket{\psi_{SR}^{E'}}$ is motivated by the fact that for a given average energy $E'$ of the sender's subsystem, the state with the maximal local entropy (entanglement) is a pure state whose local marginal is a thermal state with an average energy equal to $E'$.  Furthermore, the protocol involves grouping
the total of \( d^2 \) unitaries into \( d \) distinct sets, denoted as \( V_k \), each comprising \( d \) unitaries. Note that, \( V_0 \) represents the passive unitary set, while all \( V_k \) for  \( k \geq 1 \) are active unitary sets which increase energy during encoding. For all encoding unitary $U_{km}\in V_k$, the energy of the marginal state of the sender is identically altered. In particular, for $\ket{\psi_{SR}} = \sum_{n=0}^{d-1} \sqrt{p_n^{E'}}\ket{n,n}$, we have 
\begin{eqnarray}
    \text{Tr}_S \big(\rho_S^{km} H \big) = \sum_{n=0}^{d-1} (n+k ~\text{mod}  ~d)p_n^{E'} =:  E'_k ~~\forall m,
\end{eqnarray}
where $\rho_S^{km} = \text{Tr}_R \big( (U_{km} \otimes \mathbb I)\ket{ \psi_{SR}^{E'}} \big)$. 
Lastly, for $E < \mathbb E_d$ the expression of the capacity can be obtained via a constrained optimization over $d$ parameters, rendering the dimension of the parameter space to $\mathcal{O}(d)$.
Now we are set to present the achievable capacity by employing Protocol.~\ref{protocol:ecdc}.

\begin{theorem}
\label{th:ecdcc}
 The average energy-constrained DC capacity can be compactly expressed as
\begin{eqnarray}
              C_{\{E\}_{\emph{A}}}^{{\tt DC}} =  \Bigg\{
        \begin{array}{cc}
           2 \log_2d   & E \geq \mathbb E_d \\ 
       \mathcal{C}(E)  & E < \mathbb E_d,
        \end{array}
\end{eqnarray}
where
\begin{equation}
\begin{aligned}
\mathcal C(E)=  & \max_{\{q_k\}_{k=0}^{d-1},E^\prime} H(\{q_k\}) + {\tt Ent}\Big(\ket{\psi_{SR}^{E'}}\Big)\\
\textrm{subject to}  \quad & \sum_{k=0}^{d-1} q_k E_k' \leq E,
\end{aligned}
\label{eq:complementary}
\end{equation}  
where $E_k' = \sum_{n=0}^{d-1} (n+k ~\emph{mod}  ~d)p_n^{E'}$, with $E'_0 = E'$. The form of $\ket{\psi_{SR}^{E'}}$ is given in Protocol.~\ref{protocol:ecdc}.
\end{theorem}
\noindent The detailed proof is furnished in Appendix \ref{app:proofDCd}. Note that the DC capacity without the energy constraint also reads  \(C(|\psi_{SR}\rangle)= \log_{2}d + \tt Ent (|\psi_{SR}\rangle)\) where \(d\) is the dimension of the sender's side.

Since $V_0$ is the set of passive unitaries, the value of $H(\{q_k\})$ with $q_0\neq 1$ gauges the proportion with which the active unitaries are employed in the encoding process. The fact that $\mathcal{C}(E)$ can be expressed as the sum of an activity component $H(\{q_k\})$, and the entanglement entropy of the shared state ${\tt Ent}\Big(\ket{\psi_{SR}^{E'}}\Big)$, the maximization in Eq.~\eqref{eq:complementary} aims to achieve an optimal balance (trade-off) of shared entanglement and the activity of the encoding unitaries such that the energy bound is respected. Notice that the capacity-maximizing shared state possesses entanglement strictly less than what the given energy bound supports. It follows from the observation that when one considers $\ket{\psi_{SR}^{E'=E}}$, it forces all the encoding unitaries to be passive, i.e., $H(\{q_k\}) =0$. This immediately compels $\mathcal{C}(E) = {\tt Ent}\Big(\ket{\psi_{SR}^{E'=E}}\Big) = H(\{p_n^E\})$, i.e., the classical bound (see Eq. \eqref{eq:eccc3}), thereby wiping out the possibility of obtaining any quantum advantage. Since the entanglement of the optimal state is less than what the given energy bound supports maximally, we dub this phenomenon  \emph{more capacity with less entanglement}.


Finally, a crucial question remains whether the Protocol.~\ref{protocol:ecdc} is optimal for achieving the energy-constrained dense coding capacity in Eq.~\eqref{eq:dccc}. We address this for $d=2$ as the local dimension of the sender and receiver, in the subsequent section. 


\subsection{Energy-constrained dense coding for two-qubits}

The mathematical simplifications on offer in $d = 2$ allow us to probe energy-constrained DC in greater detail. 
Before presenting our results for the optimality of Protocol. \ref{protocol:ecdc} for $d=2$, we first furnish an important no-go result. 
\begin{theorem}
    No energy-preserving (passive) operation can provide any quantum advantage for energy-constrained entanglement-assisted capacity for $d=2$. 
    \label{theorem:passive}
\end{theorem}
Although details of the proof are provided in the Appendix~\ref{app:passive}, we provide a brief sketch of the proof here. From Lemma. \ref{lemma:epc}, we notice that the encoding quantum operations are restricted to mixtures of energy-preserving unitaries whose support is constrained to $\{ \mathbb I, \sigma_z\}.$ Encoding with such operations restricts the encoded ensemble of states in a two-dimensional subspace. This observation, along with the energy constraint on the transmitted (marginal) states from the sender, limits any possible option of quantum advantage with the assistance of pre-shared entanglement.

Now we address the issue of optimality of Protocol.~\ref{protocol:ecdc} in the noiseless case for $d=2$.
\begin{proposition}
    The energy-constrained dense coding Protocol.~\ref{protocol:ecdc} is optimal for $d = 2$.
    \label{prop:optimaldcd=2}
\end{proposition}
\noindent The details of the proof are provided in Appendix~\ref{app:proofofoptimaldcd=2}. We also calculate the exact expression of DC capacity as 
\begin{eqnarray}
    C_{\{E\}_{\emph{A}}}^{{\tt DC}}=2H\left( \frac12 \big(1-\sqrt{1-2E} \big)\right)~~\text{for }E<\frac{1}{2},
\end{eqnarray}
achieved by optimal value of $q_0=1-E'$ where $E'=(1-\sqrt{1-2E})/2$ (see Appendix~\ref{app:proofofoptimaldcd=2}). Physically, the optimal shared state for an average energy bound $E<1/2$ is the pure state of the form given in Protocol.~\ref{protocol:ecdc} with entanglement (as measured by concurrence \cite{Wootters1998}) ${\tt Con}(\ket{\psi_{SR}^{E'}}) = \sqrt{2E}$.   Notice that the quantum advantage, expressed as \( C^{\tt DC}_{\{E\}_{\emph{A}}}/C_{\{E\}} \) 
grows monotonically with $E \in (0,\frac12)$ and attains its maximal value $2$ at $E = \frac12$, and remains so for all $E \geq \frac12$.
Similar analysis can also be carried out for the cases of equiprobable messages and strict energy constraints in noiseless and dephasing scenarios (see Appendices.~\ref{app:onc} and \ref{app:cdc}).

\subsection{Enhancement of energy-constrained classical capacity of CQ-channels with entanglement assistance}
\label{subsec:shirokov}
In the unconstrained scenario, it was shown \cite{Shirokov2012}  that entanglement-assistance fails to provide any advantage to the classical capacity of CQ channels \cite{Ruskai2003, Horodecki2003}. A CQ channel $\Lambda^{\text{CQ}}(\rho) = \sum_k \zeta_k \bra{\psi_k} \rho \ket{\psi_k},$ is an entanglement-breaking channel where $\{\ketbra{\psi_k}{\psi_k} \}$ constitutes a set of projectors, and $\zeta_k$s are density matrices.
Now, an important question arises: does the equivalency of classical capacity of the CQ channel between the entanglement-assisted case and the unassisted ones hold in the energy-constrained setting?  We answer the question negatively. 
\begin{figure}[h]
        \centering
        \includegraphics[width=\linewidth]{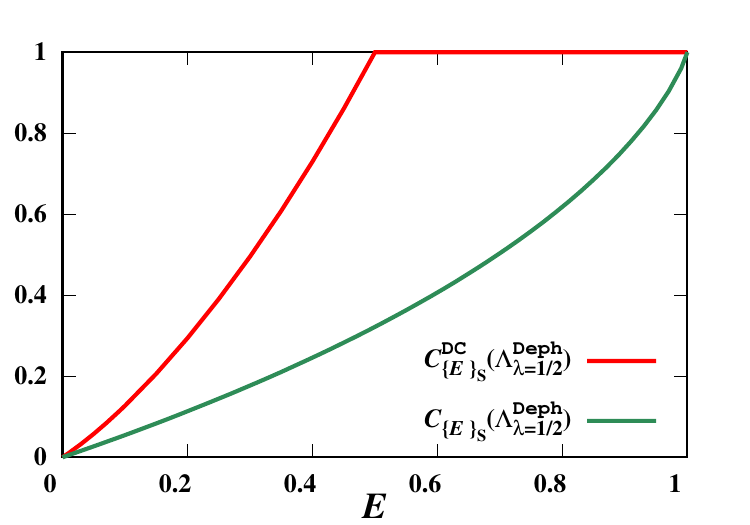}
        \caption{Strict energy-constrained DC capacity $C^{\tt DC}_{\{E\}_S}(\Lambda_{\lambda=\frac{1}{2}})$ (red) and classical capacity $C_{\{E\}_S}(\Lambda_{\lambda=\frac{1}{2}})$ (green) in the case of complete dephasing qubit channel. In contrast to the unconstrained case \cite{Shirokov2012}, entanglement provides an advantage in the classical capacity of a CQ channel in the strict energy-constrained scenario. Both axes are dimensionless.}
        \label{fig:shirokov1}
    \end{figure}
In particular, for  $\Lambda^{\tt Deph}_{\lambda = \frac{1}{2}}$, which is a CQ  channel having effect on arbitrary state $\rho$ as $\Lambda^{\tt Deph}_{\lambda = \frac{1}{2}}(\rho) = \ketbra{0}{0} \bra{0}\rho\ket{0} + \ketbra{1}{1} \bra{1}\rho\ket{1}$, we demonstrate that the entanglement-assisted DC capacity satisfies $C^{\tt DC}_{\{E\}_{\text{S}}}(\Lambda^{\tt Deph}_{\lambda = \frac{1}{2}}) > C_{\{E\}_{\text{S}}}(\Lambda^{\tt Deph}_{\lambda = \frac{1}{2}})$ for all $E \in (0,1)$ (see Fig.~\ref{fig:shirokov1}). This guarantees a finite advantage offered by entanglement assistance over the unassisted case with energy constraints. See Appendix~\ref{app:cdc} for more details, including the discussions of the feature on the average (strict) energy constraint with equiprobable messages.
The above result again strengthens our decision to choose $\Lambda^{\tt Deph}_{\lambda}$ as the prototypical example of noise in our analysis. 



\section{Conclusion}
\label{sec:conclusion}

Owing to the information transmission with limited energy transfer or purity, this work takes a step toward resolving the incongruity between the rigorous mathematical formalism of capacity and the physically relevant communication models. Our investigation unveils the emergence of characteristic energy scales below which the unconstrained framework of classical capacity does not work. Apart from obtaining compact forms of energy-constrained classical capacities for noiseless quantum channels, we found the classical capacities of dephasing quantum channels in which the probabilities in the optimal ensemble are energy-dependent. 




We developed the energy-constrained entanglement-assisted protocol, which was also found to be optimal for two-qubit systems. Surprisingly, we demonstrated that under energy constraints, classical-quantum channels display a higher classical capacity when entanglement is apriori shared between the sender and receiver -- a stark contrast to the unconstrained scenario where no such gain is found. This shows how physical limits can profoundly alter the terrain of quantum advantage in communication. In the future, it would be interesting to expand this formalism to other channels, such as energy-preserving ones and general encoding processes (in the case of entanglement assistance).



\section*{Acknowledgements}
We acknowledge the use of \href{https://github.com/titaschanda/QIClib}{QIClib} -- a modern C++ library for general purpose quantum information processing and quantum computing (\url{https://titaschanda.github.io/QIClib}) and cluster computing facility at Harish-Chandra Research Institute. PH acknowledges ``INFOSYS
scholarship for senior students".

\bibliography{bib}
   
\newpage
\newpage
\appendix

\onecolumngrid
\section{Proof of Theorem. \ref{theorem:eccp1}}
\label{app:1}
\noindent The proof proceeds in two steps -- we derive an upper bound of the energy-constrained classical capacity and then show that the upper bound can be achieved via encoding.
\begin{lemma}
    The classical capacity under average energy constraint $E$ is bounded by
    \begin{eqnarray}
        C_{\{E\}} :  \Bigg\{
        \begin{array}{cc}
          = \log_2 d   & E \geq \mathbb E_d \\ 
           \leq S(\tau_E)  & E < \mathbb E_d,
        \end{array}
        \label{eq:11_a}
    \end{eqnarray}
    where $\mathbb E_d = \frac{d-1}{2}$, and $\tau_E$ is the $d$-dimensional thermal state. Here $S(\rho)$ denotes the von-Neumann entropy of $\rho$.
    \label{th:1}
\end{lemma}
\begin{proof}
   The central quantity of interest is the Holevo quantity $\chi(\{p_i,\rho_i\})$,
    \begin{eqnarray}
        \chi(\{p_i,\rho_i\}) &=& S(\bar{\rho}) - \sum_{i} p_i S(\rho_i) \leq S(\bar{\rho}). 
    \end{eqnarray}
    where \(\bar\rho=\sum p_{i}\rho_{i}\). Note that since \(S(\rho_i)\geq0\), \(\chi(\{p_i,\rho_i\})\leq S(\bar\rho)\) which is trivially attained for pure state ensembles.
    For the energy constrained classical capacity in Eq. (\ref{eq:cc1}),
    The maximal value of the upper bound is achieved when $S(\bar{\rho})$ is maximized with respect to the energy bound Tr$(\bar{\rho}H) \leq E$. Two distinct cases arise during this maximization. 

For any value of the energy constraint $E \geq \mathbb E_d$, we note that $S(\bar{\rho})$ is maximized when $\bar{\rho}$ is the maximally mixed state $\bar{\rho} = \mathbb{I}/d$. The energy of the maximally mixed state is $ \frac{1}{d} \text{Tr}(H) = \mathbb E_d.$ The capacity of $\log_2 d$ can be achieved by encoding the $k^{\text{th}}$ classical message in the state 
\begin{eqnarray}
  \ket{\psi_k} = \frac{1}{\sqrt{d}} \sum_{n=0}^{d-1} \omega^{kn} \ket{n},  
  \label{eq:psi_k}
\end{eqnarray}
 where $\{\ket{\psi_k} \}_{k=0}^{d-1}$ forms the quantum Fourier transform (QFT) of the Hamiltonian eigenbasis $\{ \ket n \}$, and $\omega = e^{\frac{2\pi i}{d}}$ is the $d^{\text{th}}$ root of unity with $i=\sqrt{-1}$.

When $E < \mathbb E_d$, $S(\bar{\rho})$ is maximized when $\bar{\rho}$ is a thermal (Gibbs) state with energy $E$. This directly follows from the fact that the state which maximizes entropy for a given energy is the thermal state $\tau_E = e^{-\beta_E H}/Z_E$, where $Z_E = $Tr$(e^{-\beta_E H})$ satisfying $\text{Tr}(\tau_EH) = E.$
        Therefore, for $E < \mathbb E_d$, we have $C_{\{E\}} \leq S(\tau_E).$ This completes the second part of the proof.
    \end{proof}

Let us now illustrate that the upper bound of $C_{\{E<\mathbb E_d\}}$ in Eq. (\ref{eq:11_a}) can be achieved via suitable encoding. For $E < \mathbb E_d$, the $k^{\text{th}}$ message (appearing with probability $1/d$) is encoded in the state
       \begin{eqnarray}
           \ket{\psi_k} = \sum_{n=0}^{d-1} \omega^{kn} \sqrt{p_n^E}\ket{n}.
           \label{eq:psi_k_ec}
       \end{eqnarray}
       where 
    the probabilities $\{p_l^E\}_{l=0}^{d-1}$ are the weights of the thermal state, $\tau_E = \sum_{n=0}^{d-1}p_n^E \ketbra{n}{n}$. 
     Now, we have
     \begin{eqnarray}
      \bar\rho =    \frac1d \sum_{k=0}^{d-1} \ketbra{\psi_k}{\psi_k} = \tau_E.
     \end{eqnarray}
     All off-diagonal terms vanish identically, as each term is proportional to the sum $1 + \omega + \omega^2 + \ldots + \omega^{d-1}=0$. Noting that $S(\bar\rho) = S(\tau_E) = H(\{p_k^E\})$, it completes the proof of saturation.

The capacities remain the same when one considers the strict energy constraint as well, since, the energy of the encoding states, $\ket{\psi_k}$s in Eqs. ~\eqref{eq:psi_k} and \eqref{eq:psi_k_ec} corresponding to the entire range of the energy constraint $E$  are exactly equal to $\sum_{n=0}^{d-1} n p_n^E = E$. 
This completes the proof of Theorem.~\ref{theorem:eccp1}.

 \subsection{Special cases: $d = 2,3$}
 For \(d=2\)and \(3\), let us present explicitly jot down the forms of the encoding states.
In  $d =2$, for an energy (strict) constraint $E < \mathbb E_{d=2} = 1/2$, the following states are considered for encoding the binary message:
\begin{eqnarray}
    \ket{\psi_0} &=& \sqrt{1-E} \ket{0} + \sqrt{E} \ket{1}, \nonumber \nonumber \\ \textrm{and} \quad 
     \ket{{\psi}_1} & = & \sqrt{1-E} \ket{0} - \sqrt{E} \ket{1}.
     \label{eq:qubit_encoding_AE}
\end{eqnarray}
From Theorem. \ref{theorem:eccp1}, we have a compact expression for the energy-constrained classical capacity in the $d = 2$ case.
\begin{eqnarray}
          C_{\{E\}} =  \Bigg\{
        \begin{array}{cc}
           1   & E \geq \frac12 \\ 
           H(E)  & E < \frac{1}{2}.
        \end{array}
        \label{eq:capacityd=2}
    \end{eqnarray}

For $d=3$, the states employed to encode the three messages for $E < 1$ can be represented as  
\begin{eqnarray}
    \ket{\psi_0} &=& \sqrt{p_0^E}\ket0+\sqrt{p_1^E}\ket1+\sqrt{p_2^E}\ket2, \nonumber \\
    \ket{\psi_1} &=& \sqrt{p_0^E}\ket0+ \omega\sqrt{p_1^E}\ket1+\omega^2\sqrt{p_2^E}\ket2, \nonumber \\ \textrm{and} \quad 
    \ket{\psi_2} &=& \sqrt{p_0^E}\ket0+ \omega^2\sqrt{p_1^E}\ket1+\omega\sqrt{p_2^E}\ket2,
\end{eqnarray}
where $\omega = e^{\frac{2\pi i}{3}}$ is one of the cubic root of unity. The corresponding probabilities read as
\begin{eqnarray}
    p_0^E &=& \frac16(7-3E-\sqrt{1+6E-3E^2}), \nonumber\\
    p_1^E &=& 2-2p_0^E - E, \nonumber \\ \textrm{and}  \quad 
    p_2^E &= &1 - p_0^E - p_1^E.
\end{eqnarray}
From $d = 3$ onwards, we have a closed-form expression of the capacity and the encoding states, although the closed form of probabilities in terms of energy become cumbersome.

\begin{figure}[h]
\includegraphics [width=0.5\linewidth]{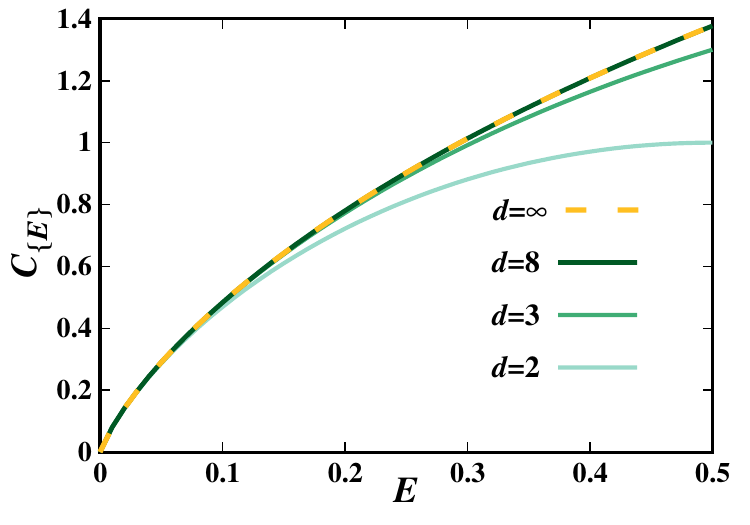}
\caption{The energy-constrained classical capacity, \(C_{\{E\}}\) (vertical axis), plotted against the average energy constraint, \(E\) (horizontal axis), for different dimensions \(d\). As the dimension increases, the value of \(C_{\{E\}}\) also increases. Dimensional advantage is significant when $E$ is high, although it saturates rapidly with increasing $d$. Both axes are dimensionless.} 
\label{fig:cc_avgE_d}
\end{figure}

\subsection{Dimensional advantage}
\label{app:dimadv}
This energy scale allows us to track the dimensional advantage of energy-constrained capacity for $E < \mathbb E_d$. If $E < \mathbb E_{d=2} = \frac12$, we compute $C_{\{E\}}$ for various dimensions $d$, using techniques of statistical mechanics. To begin with we evaluate the partition function $Z_E = \sum_{n=0}^{d-1} e^{-n \beta_E} = \frac{1-e^{-d\beta}}{1-e^{-\beta}}$.  Now, recall that $ E = -\frac{\partial}{\partial \beta} \ln Z_E = \frac{1}{e^\beta-1} - \frac{d}{e^{d\beta}-1}$ \cite{reif2010statistical}. 
Formally, the solution of the above equation $\beta^*_d$ for arbitrary $d$ can be expressed as $  \beta^*_d = \beta(d,E).$
For $d = \infty$, we have $\beta^*_{d = \infty} = \ln (1 + E^{-1})$. For finite $d$, we find out $\beta^*_d$ via standard root finding algorithms \cite{press2007numerical}.
Substituting this value of $\beta^*_d$ in the expression of the entropy in terms of average energy and the partition function, we get $  C_{\{E\}} = (\beta^*_d E + \ln Z_E)\log_2 e.$
Note that for $d = \infty$,  $C_{\{E\}} = (1+E)\log_2 (1+E) - E \log_2 E$. We plot $C_{\{E\}}$ vs. $E$ for various $d$ in Fig.~\ref{fig:cc_avgE_d}, to explicitly demonstrate the dimensional advantage of $C_{\{E\}}$ even in the presence of energy constraint.

\section{Proof of Lemma. \ref{lemma:epc}}
\label{app:lemma1}
The proof proceeds by noting a result from Ref. \cite{Chiribella2017} which proves that all energy-preserving channels are unital. Also, note that all unital qubit channels can be expressed as mixtures of unitaries \cite{Kmmerer1987}. Therefore, the most general qubit energy-preserving channel is
\begin{eqnarray}
    \Lambda^{\tt EP}(\rho) = \sum_i p_i \mathcal{U}_i(\rho)
\end{eqnarray}
where $\mathcal{U}_k(\sigma) \equiv U_k \sigma U_k^\dagger$. The property of energy preservability induces a stronger requirement on $U_k$s. Now, the channel $\Lambda^{\tt EP} \equiv \{p_i, U_i\}$ is energy preserving iff
\begin{eqnarray}
    \text{Tr} (\Lambda^{\tt EP}(\rho)H) &=& \text{Tr} (\rho H) ~~\forall \rho \nonumber \\
    \implies \sum_i p_i \text{Tr}(\mathcal{U}_i(\rho)H) &=& \text{Tr} (\rho H) ~~\forall \rho \nonumber \\ 
    \implies \sum_i p_i \text{Tr}(\rho ~\mathcal{U}_i^\dagger(H)) &=& \text{Tr} (\rho H) ~~\forall \rho \nonumber \\
    \implies  \text{Tr}(\rho ~\sum_i p_i\mathcal{U}_i^\dagger(H)) &=& \text{Tr} (\rho H) ~~\forall \rho \nonumber \\
    \implies \sum_i p_i \mathcal{U}_i^\dagger(H) &=&  H  \nonumber \\
    \implies \sum_i p_i \ketbra{\psi_i}{\psi_i} &=& \ketbra{1}{1},
    \label{eq:ifff}
\end{eqnarray}
where in $d=2$, we have $H = \ketbra{1}{1}$, and $\ket{\psi_i}=U_i^\dagger\ket 1$. Then the iff condition in Eq. \eqref{eq:ifff} translates to
\begin{eqnarray}
    \sum_i p_i \ketbra{\psi_i}{\psi_i} = \ketbra{1}{1}.
    \label{eq:iff2}
\end{eqnarray}
Since the R.H.S is a pure state (extreme point of set of states), Eq. \eqref{eq:iff2} holds iff $\ket{\psi_i} \propto \ket 1$ for all $i$. It right away implies, for all $i$, 
\begin{eqnarray}
    U_i^\dagger \ketbra{1}{1} U_i = \ketbra{1}{1} \implies \ketbra{1}{1} U_i = U_i \ketbra{1}{1} \implies U_i H -  H U_i = 0 \implies [U_i,H] = 0.
\end{eqnarray}
This completes the proof.

\section{Capacity of the dephasing channel}
\label{app:3}
Here, we derive the optimality of the encoding states and the capacity of the dephasing channel ($\Lambda^{\tt{Deph}}_\lambda$) in the energy-constrained scenario, mentioned in the main text.

\subsection{Average energy constraint}
To gain insight into the structure of equiprobable optimal signal states under arbitrary dephasing noise, subject to the constraint that the average ensemble energy satisfies \( E <  \frac{1}{2} \), we begin by examining the limiting case of a completely dephasing channel, namely when \( \lambda = \frac{1}{2} \). In this scenario, since the Bloch sphere is reduced to the $z$-axis only, the optimal encoding is performed by $\Lambda^{\tt{Deph}}_{\lambda=\frac12}\left({\tt Proj}(\ket0)\right)$ and $\Lambda^{\tt{Deph}}_{\lambda=\frac12}\left({\tt Proj}(\sqrt{1-2E}\ket0+\sqrt{2E}\ket1)\right)$. This, in turn, gives $\tilde C_{\{E\}_A}(\Lambda^{\tt{Deph}}_{\lambda=\frac12})=H(E)-H(2E)/2$. Having discussed the optimal signal states in both limiting cases of $\lambda=1/2$ and $\lambda=1$ (i.e., the noiseless scenario, described in Theorem.~1 and Eq.~\eqref{eq:qubit_encoding_AE}), we can write down the Holevo capacity $\forall \lambda$ in the region $E<1/2$ as 
\begin{eqnarray}
    \widetilde C_{\{E\}_A}(\Lambda^{\tt Deph}_{\lambda}) = \max_{\Delta \in [0,E]}\chi\left(\{1/2,\Lambda_\lambda^{\tt Deph}(\rho_\pm(E,\Delta))\}\right),
    \label{eq:lambda_ec_cc}
\end{eqnarray}
where $\rho_{\pm}={\tt Proj}\left(\sqrt{1-E_\pm}\ket0\pm\sqrt{E_\pm}\ket1\right)$ are used for encoding with $E_\pm=E\pm\Delta$. Note that the optimization in Eq.~\eqref{eq:lambda_ec_cc} is over a single parameter $\Delta(E,\lambda)$ only. In limiting cases we have $\Delta(E,1/2)=E$ and  $\Delta(E,1)=0$ as argued previously. We calculate $\widetilde C_{\{E\}_A}(\Lambda^{\tt Deph}_{\lambda})$ for various noise strength $\lambda$ and energy bound $E$ as depicted in Fig.~\ref{fig:cc_avg_strict}(a).
\begin{figure*}
\includegraphics[width=\textwidth]{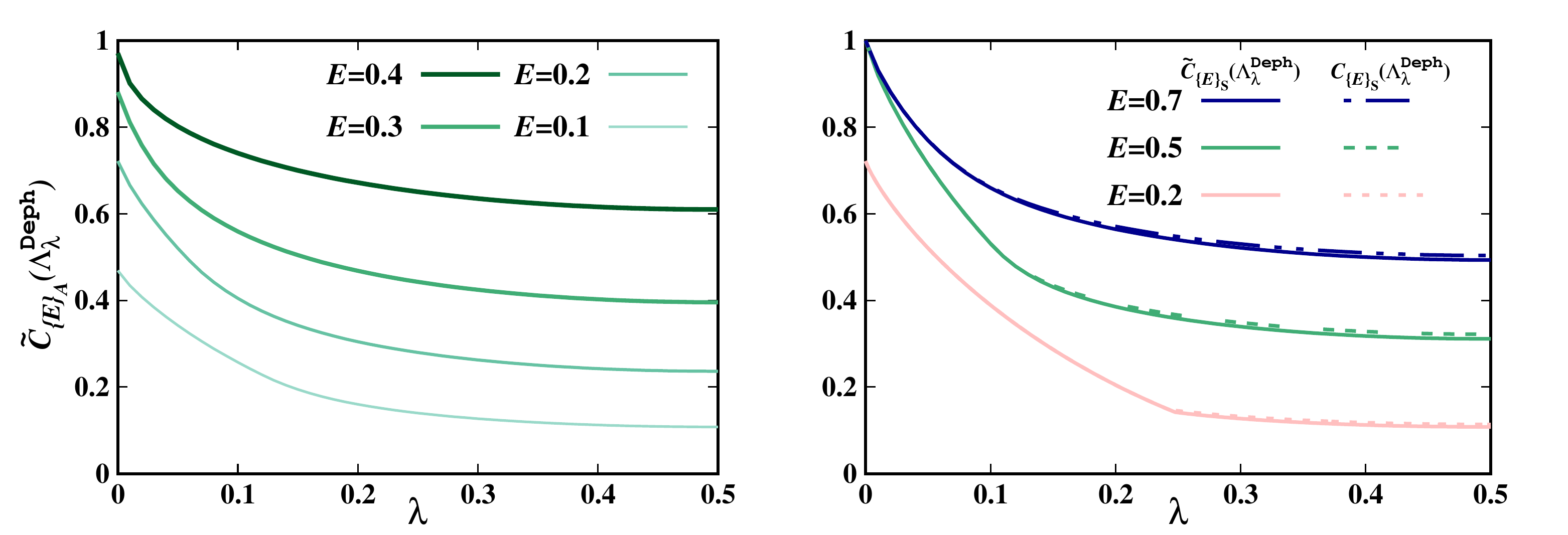}
\caption{(a) Average energy-constrained classical capacity, $\widetilde{C}_{\{E\}_A}(\Lambda^{\mathrm{Deph}}_{\lambda})$ (vertical axis) against the dephasing noise parameter, $\lambda$ (horizontal axis), for various values of the average energy constraint, $E$. (b) Strict energy-constrained classical capacity (vertical axis) versus the dephasing noise parameter, \(\lambda\) (horizontal axis), for different values of the energy bound \(E\). Here, \(\widetilde{C}_{\{E\}_{S}}\big(\Lambda^{\mathrm{Deph}}_{\lambda}\big)\) (solid-lines) represents the capacity under an equiprobable input ensemble while \(C_{\{E\}_{S}}\big(\Lambda^{\mathrm{Deph}}_{\lambda}\big)\) (dashed-lines) corresponds to the optimal probability scenario. It is observed that the difference \(C_{\{E\}_{S}}\big(\Lambda^{\mathrm{Deph}}_{\lambda}\big) - \widetilde{C}_{\{E\}_{S}}\big(\Lambda^{\mathrm{Deph}}_{\lambda}\big)\geq\) \(10\sigma\) within the interval \(0.4 \leq \lambda \leq 0.5\) for \(E = 0.5\) and \(E = 0.7\), where \(\sigma = 5 \times 10^{-5}\). As \(E\) increases, the curves shift from bottom to top. Both axes are dimensionless.
} 
\label{fig:cc_avg_strict}
\end{figure*}

Contrastingly, when $E\geq 1/2$, the center of the Bloch sphere is available and encoding can be done with $\ket0$ and $\ket1$ states to get $\widetilde C_{\{E\}_A}(\Lambda^{\tt Deph}_{\lambda})=1$.

\subsection{Strict energy constraint}
After determining the optimal signal ensemble under the average EC scenario, we observe that when a strict EC is imposed, the capacity with equiprobable messages can be expressed as 
\begin{eqnarray}
    \widetilde C_{\{E\}_S}(\Lambda^{\tt Deph}_{\lambda}) = \max_{\Delta \in [0,E]}\chi\left(\{1/2,\Lambda_\lambda^{\tt Deph}(\rho_0(E));1/2,\Lambda_\lambda^{\tt Deph}(\rho_-(E,\Delta))\}\right),
\end{eqnarray}
with $\rho_0={\tt Proj}(\sqrt{1-E}\ket0+\sqrt{E}\ket1)$ and $\rho_-={\tt Proj}(\sqrt{1-E_-}\ket0-\sqrt{E_-}\ket1)$. In the case of complete dephasing channel, the capacity simplifies to $\widetilde C_{\{E\}_S}(\Lambda^{\tt Deph}_{\lambda=\frac12})=H(E/2)-H(E)/2$, since  $\Delta(E,\frac12)=E$ in this regime. Note that, for the noiseless case, $\Delta(E,1)=0$. For other nontrivial values of $\lambda$, we resort to numerical optimization over $\Delta$.

Interestingly, when the whole ensemble is optimized, although the encoding states remain the same as $\rho_0$ and $\rho_-$, the probability $p_0$ becomes a function of energy bound, i.e., 
\begin{eqnarray}
    C_{\{E\}_S}(\Lambda^{\tt Deph}_{\lambda}) = \max_{p_0,\Delta \in [0,E]}\chi\left(\{p_0(E),\Lambda_\lambda^{\tt Deph}(\rho_0(E));(1-p_0(E)),\Lambda_\lambda^{\tt Deph}(\rho_-(E,\Delta))\}\right),
\end{eqnarray}
which requires two-parameter optimization. In Fig.~\ref{fig:cc_avg_strict}(b), we present a comparison between \( C_{\{E\}_S}(\Lambda^{\tt{Deph}}_{\lambda}) \) and \( \widetilde{C}_{\{E\}_S}(\Lambda^{\tt{Deph}}_{\lambda}) \) across various noise regimes under different values of fixed energy constraint \( E \). From this analysis, we conclude that \( C_{\{E\}_S}(\Lambda^{\tt{Deph}}_{\lambda}) \geq \widetilde{C}_{\{E\}_S}(\Lambda^{\tt{Deph}}_{\lambda}) \). Note that equality holds in two limiting cases: $(1)$ when \( E = 1 \) for any value of \( \lambda \), and $(2)$ when \( \lambda = 0 \) or \( \lambda = 1 \), for any value of the energy bound \( E \).

\section{Classical capacities under dual constraints of energy and purity}
\label{app:dual}

We now move on to a setting, in which apart from energy, another physical property, namely purity of the states used for encoding is bounded. From Eq.~\eqref{eq:cc1}, we have
$N_A = N_S =1$, where $\mathcal{P}_1(\sigma) = $ Tr$(\sigma H)$, and $\mathcal{D}_1(\sigma) =$ Tr$(\sigma^2)$, for any state $\sigma$. With the above identifications, we are set to compute the classical capacity of a quantum channel $\Lambda$ with dual constraints of average energy and strict purity $C_{\{E\}_{\text{A}}\{L\}_{\text{S}}}(\Lambda).$

The capacity of a noiseless qubit channel turns out to be 
\begin{eqnarray}
    C_{\{E\}_{\text{A}}\{L\}_{\text{S}}} = \Bigg\{ 
        \begin{array}{cc}
           1 - H(r_L^-)   & E \geq \frac12, \\ 
            H(E)-H(r_L^-)  & r_L^- \leq E < \frac12, \\ 
            0 & E < r_L^-,
        \end{array}
        \label{eq:avgE_purity}
\end{eqnarray}
where like in Eq. \eqref{eq:eccc3}, we use $C_{\{E\}_{\text{A}}\{L\}_{\text{S}}} := C_{\{E\}_{\text{A}}\{L\}_{\text{S}}}(\Lambda = \mathbb I)$, and $r_L^{\mp} = \frac{1 \mp r_L}{2}$, where $r_L = \sqrt{2L-1}$ is the Bloch radius up to which states can be used for encoding consistent with the strict purity constraint. 
From the definition of states, we have $L \geq \frac12$, guaranteeing $r_L \geq 0$.
The details of the calculation are furnished in Appendix \ref{subsec:noiseless}.  Interestingly, the purity constraint induces another energy scale in the problem on top of $E = \frac12$, namely $r_L^-$. When the energy bound $E < r_L^-$, the set of available states for encoding becomes null and hence the capacity vanishes.

 Like in the case of only average energy constraint, the noisy capacity also matches with the noiseless one: $C_{\{E\}_{\text{A}}\{L\}_{\text{S}}}(\Lambda^{\tt Deph}_\lambda) = C_{\{E\}_{\text{A}}\{L\}_{\text{S}}}$. Again, a crucial difference between noisy and noiseless case occurs from the optimal encoding states and probabilities, specifically the optimal value of the noisy capacity is achieved by $\{1-p ,\Lambda^{\tt Deph}_\lambda(\rho_+), p, \Lambda^{\tt Deph}_\lambda(\rho_-) \}$, where $\Lambda^{\tt Deph}_\lambda(\rho_\pm) = \rho_\pm = r_L^\mp \ketbra{0}{0}+ r_L^\pm \ketbra{1}{1}$, and
\begin{eqnarray}
    p = \Bigg\{ 
        \begin{array}{cc}
          \frac12  & E \geq \frac12 \\ 
           \frac12 \Big( 1 + \frac{1-2E}{r_L} \Big)  & r_L^- \leq E < \frac12. 
        \end{array}
\end{eqnarray}
Notably, the probability of encoding for maximizing the capacity depends on both the energy and the purity constraints. While in the noiseless case, the encoding probabilities can be chosen as equiprobable irrespective of the energy and purity constraints.

In case of independent sources requiring the classical messages to come with equal probabilities,
let us consider the impact of this setting on the capacity of the dephasing channel, e.g, when $\lambda = \frac12$. Our computations reveal

\begin{eqnarray} \widetilde{C}_{\{E\}_{\text{A}}\{L\}_{\text{S}}}\big(\Lambda^{\tt Deph}_{\lambda=\frac12}\big) &=&\Bigg\{
      \begin{array}{cc}
           1 - H(r_L^-)   &E\geq \frac12, \\ 
            H(E)-\frac12 \big(H(r_L^-) +  H(2E-r_L^-)\big)  &  r_L^- \leq E < \frac12, \\ 
            0 &  E < r_L^-.
        \end{array}
        \label{eq:noise_avgE_purity}
\end{eqnarray}
Details of the calculation are provided in Appendix \ref{subsec:noisy}.

Here, we compute the optimal encoding states and corresponding capacities under energy and purity constraints for noiseless and noisy dephasing channels.
\subsection{Noiseless channel}
\label{subsec:noiseless}
In the noiseless scenario, to understand the structure of optimal signal states for encoding under energy and strict purity constraints, we consider the following setup. The average energy constraint requires that \( E \geq \frac{1}{2} \), while the strict purity constraint restricts the set of allowable states to lie on the surface of a Bloch sphere of radius \(r_L = \sqrt{2L - 1}\)
where \( L\) is the purity bound.

To simplify the analysis, we consider two extremal states aligned along the vertical (z-) axis of the Bloch sphere. These encoded states are defined as \(
\rho_{\pm} = r_L^{\mp} \ket{0}\bra{0} + r_L^{\pm} \ket{1}\bra{1},
\) where \( r_L^{\mp} = \frac{1 \mp r_L}{2}.\)
These states lie on the surface of the Bloch sphere with radius \( r_L \), ensuring that they satisfy the strict purity constraint. Notably, since \( \rho_+ \) and \( \rho_- \) are symmetric about the equator of the Bloch sphere, they possess equal von Neumann entropy, i.e, \( S(\rho_+) = S(\rho_-)\) and hence have the average state, \(\bar\rho=\frac{I}{2}\). Thus, the capacity takes the form as 
\begin{eqnarray}
C_{\{E\}_{\text{A}}\{L\}_{\text{S}}}=1-H(r_L^-),   
\end{eqnarray}
 where \(H(r_L^-)=S(\rho_+)=-r_L^- \log_2(r_L^-)-(1-r_L^-)\log_2(1-r_L^-)\).

Now, we consider the complementary regime where the average energy constraint satisfies \( E < \frac{1}{2} \), while the strict purity constraint remains unchanged. Within this regime, two distinct energy intervals emerge: \( r_L^- \leq E < \frac{1}{2} \) and \( E < r_L^- \). In the interval, \( r_L^- \leq E < \frac{1}{2} \), the set of permissible encoding states are those that simultaneously lie on the energy plane \( E \) and on the surface of the Bloch sphere with radius \( r_L \). Let the encoding ensemble be \( \{(1-p), \rho_1, p, \rho_2\} \), with corresponding average state, \( \bar{\rho} = (1-p)\rho_1 + p\rho_2 \). Since both \( \rho_1 \) and \( \rho_2 \) lie on the surface of the Bloch sphere of radius \( r_L \), they have the same entropy, i.e., \( S(\rho_1) = S(\rho_2) = S(\rho_+) \). The entropy of the average state \( \bar{\rho} \) is given by \( S(\bar{\rho}) = H(E) \). Thus, the capacity, in turn, reads
\begin{eqnarray}
C_{\{E\}_{\text{A}}\{L\}_{\text{S}}}=H(E)-H(r_L^-).
\label{eq:noiseless_cc_int}
\end{eqnarray}

In the other constrained energy regime, where \( E < r_L^- \), no valid encoding states are available. This is because the corresponding energy plane lies entirely outside the Bloch sphere of radius \( r_L \), making it impossible to find any quantum state that simultaneously satisfies both the energy and purity constraints. As a result, the communication capacity in this regime vanishes, i.e.,\(C_{\{E\}_{\text{A}}\{L\}_{\text{S}}} = 0.\)

\subsection{Impact of dephasing channel on capacities under dual constraints}
\label{subsec:noisy}

We investigate the impact of dephasing noise on communication capacity under dual constraints of average energy and the purity condition given in Eq.~(\ref{eq:noise_avgE_purity}).
 In particular, we focus on the special case of the complete dephasing channel with \(\lambda=\frac12\). In the interval \(E\geq\frac12\), the optimal encoding ensemble \(\{1-p,\Lambda_{\lambda}^{\text{Deph}}(\rho_+),p,\Lambda_{\lambda}^{\text{Deph}}(\rho_-)\}\). Interestingly, \(\Lambda_{\lambda}^{\text{Deph}}(\rho_\pm)=\rho_\pm\) and \(
\rho_{\pm} = r_L^{\mp} \ket{0}\bra{0} + r_L^{\pm} \ket{1}\bra{1}\). Here, the optimal encoding states are equiprobable, i.e, \(p=\frac12\). Thus, the capacity takes the same form as noiseless scenario, i.e, \(\widetilde{C}_{\{E\}_{\text{A}}\{L\}_{\text{S}}}\big(\Lambda^{\tt Deph}_{\lambda=\frac12}\big)=C_{\{E\}_{\text{A}}\{L\}_{\text{S}}}=1-H(r_L^-)\).

In the average energy-constrained regime defined by the interval \(\, r_L^- \leq E < \frac{1}{2} \,\), we focus on the equiprobable encoding scenario. The optimal ensemble in this case is given by \(\left\{ \frac{1}{2}, \rho_1, \frac{1}{2}, \rho_2 \right\}\). The first state, \(\rho_1 = \rho_- = r_L^+ |0\rangle\langle 0| + r_L^- |1\rangle\langle 1|\), is chosen to have energy \(E - \Delta\), which implies \(r_L^- = E - \Delta\) and \(r_L^+ = 1 - (E - \Delta)\). To ensure that the ensemble satisfies the average energy constraint, the second state \(\rho_2\) must carry energy \(E + \Delta\). Accordingly, we take \(\rho_2 = \big(1 - (E + \Delta)\big) |0\rangle\langle 0| + (E + \Delta) |1\rangle\langle 1|\). Thus, the encoding states can take the form as
\begin{eqnarray}
    \rho_1&=& (1-r_L^-) |0\rangle\langle 0| + r_L^- |1\rangle\langle 1|,\\ \nonumber
    \rho_2&=& (1-2E+r_L^-) |0\rangle\langle 0| + (2E-r_L^-) |1\rangle\langle 1|,
\end{eqnarray}
where \(\Delta=E-r_L^-\) and the average state, \(\bar\rho=(1-E)|0\rangle \langle0| + E|1\rangle \langle 1|\). Thus, \(S(\bar\rho)=H(E)\), \(S(\rho_1)=H(r_L^-)\) and \(S(\rho_2)=H(2E-r_L^-)\) where \(H(2E-r_L^-)=-(2E-r_L^-)\log_2(2E-r_L^-)-\big(1-(2E-r_L^-)\big)\log_2\big(1-(2E-r_L^-)\big)\). Therefore, the capacity becomes 
\begin{eqnarray}
    \widetilde{C}_{\{E\}_{\text{A}}\{L\}_{\text{S}}}\big(\Lambda^{\tt Deph}_{\lambda=\frac12}\big)&=&H(E)-\frac12\big(H(r_L^-)+H(2E-r_L^-)\big).
    \label{eq:capacity_int}
\end{eqnarray}
Now we will focus on the optimal probability encoding scenario. In the interval \(r_L^-\leq E<\frac12\), the optimal encoding ensemble \(\{1-p,\Lambda_{\lambda}^{\text{Deph}}(\rho_+),p,\Lambda_{\lambda}^{\text{Deph}}(\rho_-)\}\) where \(\Lambda_{\lambda}^{\text{Deph}}(\rho_\pm)=\rho_\pm\) and the average state takes the form as 
\begin{eqnarray}
  \bar\rho&=&p\rho_-+(1-p)\rho_+\nonumber \\ &=&\big(p(1-r_L^-)+(1-p)r_L^-\big)|0\rangle\langle 0|+\big(pr_L^-+(1-p)(1-r_L^-)\big)|1\rangle\langle1|.
  \label{eq:avg_state_dc}
\end{eqnarray}
 Again, we know that \(\bar\rho\) must lie on the energy plane \(E\), i.e,  
 \begin{eqnarray}
   \bar\rho&=&(1-E)|0\rangle\langle 0|+E|1\rangle\langle 1|.
   \label{eq:avg_state_en}
 \end{eqnarray}
 Thus, comparing Eqs. (\ref{eq:avg_state_dc}) and (\ref{eq:avg_state_en}), we can obtain 
 \begin{eqnarray}
 p(1-r_L^-)+(1-p)r_L^-&=&(1-E)\nonumber\\\nonumber
 \implies p &=& \frac{1-E-r_L^-}{1-2r_L^-},\\
 \implies p&=&\frac12\bigg(1+\frac{1-2E}{r_L}\bigg).
     \label{eq:com}
 \end{eqnarray}
 Therefore, the optimal encoding probability, \(p\) depends explicitly on the average energy constraint \(E\) and the Bloch sphere radius \(r_L\), which arises from the imposed purity constraint.

Although, \(S(\bar\rho)=H(E)\), and \(S(\rho_+)=S(\rho_-)=H(r_L^-)\). Therefore, the capacity, \({C}_{\{E\}_{\text{A}}\{L\}_{\text{S}}}\big(\Lambda^{\tt Deph}_{\lambda=\frac12}\big)=H(E)-H(r_L^-)\) which is same as shown in Eq. (\ref{eq:noiseless_cc_int}).

In the other constrained energy regime, where \( E < r_L^- \), no valid encoding states are available. Thus, \(\widetilde{C}_{\{E\}_{\text{A}}\{L\}_{\text{S}}}\big(\Lambda^{\tt Deph}_{\lambda=\frac12}\big)=0\).

\section{Proof of Theorem. \ref{th:ecdcc}}
\label{app:proofDCd}
    Following Protocol \ref{protocol:ecdc}, the classical messages $\{k,m\}_{k,m=0}^{d-1}$ are encoded as 
    \begin{eqnarray}
        \{k,m\} \to  (U_{km} \otimes  \mathbb I )\ket{\psi_{SR}^{E'}} = \sum_{n=0}^{d-1} \omega^m \sqrt{p_n^{E'}} \ket{n+k ~\text{mod} ~d, ~n}.
    \end{eqnarray}
    The average state of the encoded ensemble is given by
      \begin{eqnarray}
        \bar \rho_{SR} = \frac1d \sum_{k,m = 0}^{d-1} q_k ~ (U_{km} \otimes  \mathbb I )\ket{\psi_{SR}^{E'}}\bra{\psi_{SR}^{E'}}(U_{km}^\dagger \otimes  \mathbb I )  = \bigoplus_{k=0}^{d-1} q_k 
\Big( \sum_{n=0}^{d-1} p_n^{E'} \ketbra{n+k ~\text{mod} ~d,n}{n+k ~\text{mod} ~d,n} \Big).
    \end{eqnarray}
The block-diagonal structure of $\bar \rho$ is obtained by summing over the index $m$. The average energy transmitted by the sender is computed as
\begin{eqnarray}
    \text{Tr}(\bar \rho_S H) = \sum_{k = 0}^{d-1}q_k E_k', ~~\text{where} ~E_k' = \sum_{n=0}^{d-1} (n+k ~\text{mod}  ~d)p_n^{E'}.
\end{eqnarray}
Here $\bar \rho_{S} = \text{Tr}_R \bar \rho_{SR}$, and we have $E_0' = E'$. The average energy constraint can then be succinctly expressed as
\begin{eqnarray}
     \text{Tr}(\bar \rho_S H) \leq E \implies \sum_{k = 0}^{d-1}q_k E_k' \leq E.
     \label{eq:energybounddcgen}
\end{eqnarray}
On the other hand, the eigenvalues of $\bar \rho_{SR}$ are simply $\{q_k p_n^{E'}\}_{k,n = 0}^{d-1}$. Hence the Holevo capacity is given by
\begin{eqnarray}
    S(\bar \rho_{SR}) = -\sum_{k=0}^{d-1}\sum_{n=0}^{d-1} q_k p_n^{E'} \log_2 \big(q_k p_n^{E'}\big) = \sum_{n=0}^{d-1} p_n^{E'} \times \Big( -\sum_{k=0}^{d-1} q_k \log_2 q_k  \Big) + \sum_{k=0}^{d-1} q_k \times \Big(-\sum_{n=0}^{d-1} p_n^{E'} \log_2 p_n^{E'}\Big).
\end{eqnarray}
This allows us to write 
\begin{eqnarray}
    S(\bar \rho_{SR}) =  \Big( -\sum_{k=0}^{d-1} q_k \log_2 q_k  \Big) +  \Big(-\sum_{n=0}^{d-1} p_n^{E'} \log_2 p_n^{E'}\Big) = H(\{q_k\}) + {\tt Ent} \Big(\ket{\psi_{SR}^{E'}}\Big),
    \label{eq:DCcapgend}
\end{eqnarray}
where ${\tt Ent}(\ket \phi)$ of any pure state $\ket \phi$ denotes the entanglement entropy of $\ket \phi$.
Finally, 
combining Eqs. \eqref{eq:DCcapgend} and \eqref{eq:energybounddcgen} completes the proof of Theorem. \ref{th:ecdcc}.

\section{Proof of Theorem. \ref{theorem:passive}}
\label{app:passive}
The proof begins by investigating the form of the pre-shared entangled state that may provide any assistance to the classical capacity achievable using the qubit $(d= 2)$ channel shared between them. Suppose the communicating parties, $S$ and $R$, share the state $\rho_{SR} \in \mathbb{C}^2\otimes \mathbb{C}^d$.
Now we assume that 
the receiver, $R$, has access to the purification of the state, 
\begin{eqnarray}
\rho_{SR} \to \ket{\psi_{SRR^\prime}} \in \mathbb{C}^2\otimes \mathbb{C}^{dd^\prime}.
\end{eqnarray}
Now, $\ket{\psi_{SRR^\prime}}$ can be written in the Schmidt form
\begin{eqnarray}
    \ket{\psi_{SRR^\prime}} = \sqrt{\alpha} \ket{\eta ~\mu} + \sqrt{1-\alpha} \ket{\eta^\perp ~\mu^\perp} = U_\eta\otimes \mathbb{I} (\sqrt{\alpha} \ket{0 \mu}+ \sqrt{1-\alpha}\ket{1 \mu^\perp}),
\end{eqnarray}
where $\{\ket 0, \ket 1 \}$ are the eigenstates of the system Hamiltonian, with  $U_\eta \ket{0(1)} = \ket{\eta(\eta^\perp)}$.

From Lemma.~\ref{lemma:epc}, the most general form of passive (energy-preserving) operations in $d = 2$ is given by 
\begin{eqnarray}
    \mathcal{E}(\rho)  = \sum_j q_j \mathcal{U}_j(\rho),
\end{eqnarray}
where $\mathcal{U}_k (\sigma) = U_k \sigma U_k^\dagger$ with $U_k =\cos \theta_k ~\mathbb I +i \sin \theta_k \sigma_z$, and $\sum_k q_k = 1$. The encoding of the classical message $x$ is achieved by the following encoding $x \to \rho_x =\mathcal{E}_x \otimes \mathbb I (\ketbra{\psi_{SRR^\prime}}{\psi_{SRR^\prime}})$. Now, the form of the energy preserving channels (Lemma. \ref{lemma:epc}) constrains the set of encoded states $\{\rho_x\}$ lie in a two-dimensional subspace spanned by the vectors $\{\ket{\psi_{SRR^\prime}}, \sigma_z \otimes \mathbb I \ket{\psi_{SRR^\prime}} \}$ with \(\sigma_{i}(i=x,y,z)\) being Pauli matrices. The two-dimensionality of the encoded ensemble implies Theorem.~\ref{theorem:eccp1} for $d=2$ is applicable here. Thereby limiting the capacity to the one expressed in Eq.~\eqref{eq:capacityd=2},  wiping out the possibility of any quantum advantage. This completes the proof.

\section{Energy-constrained dense coding protocol in $d=2$}
\label{app:proofofoptimaldcd=2}
In this section, we assess the dense coding protocol for both noiseless and noisy scenario when sender's side is two-dimensional.
\subsection{Noiseless DC capacity under average energy constraint}
The elements of Protocol.~\ref{protocol:ecdc} corresponding to average energy constraint in $d=2$ for $E<1/2$ can be explicitly written as 
\begin{eqnarray}
    \nonumber\ket{\psi_{SR}^{E'}}&=&\sqrt{1-E'}\ket{00}+\sqrt{E'}\ket{11},\\
    V_0&=&\{\mathbb I,\sigma_z\};~~V_1=\{\sigma_x,\sigma_y\},
    \label{eq:elements}
\end{eqnarray}
where $E'<E$. Using Eq.~\eqref{eq:elements} in Eq.~\eqref{eq:complementary} of Theorem.~\ref{th:ecdcc}, the average energy-constrained DC capacity is given by
\begin{equation}
\begin{aligned}
    C^{\tt DC}_{\{E\}_A}&=\mathcal C(E)=\max_{q_0,E'}H(q_0)+H(E')\\
    \textrm{subject to}\quad & q_0E'+(1-q_0)(1-E')\leq E.
\end{aligned}
\label{eq:qubit_dcc}
\end{equation}
\begin{figure*}
\includegraphics[width=\textwidth]{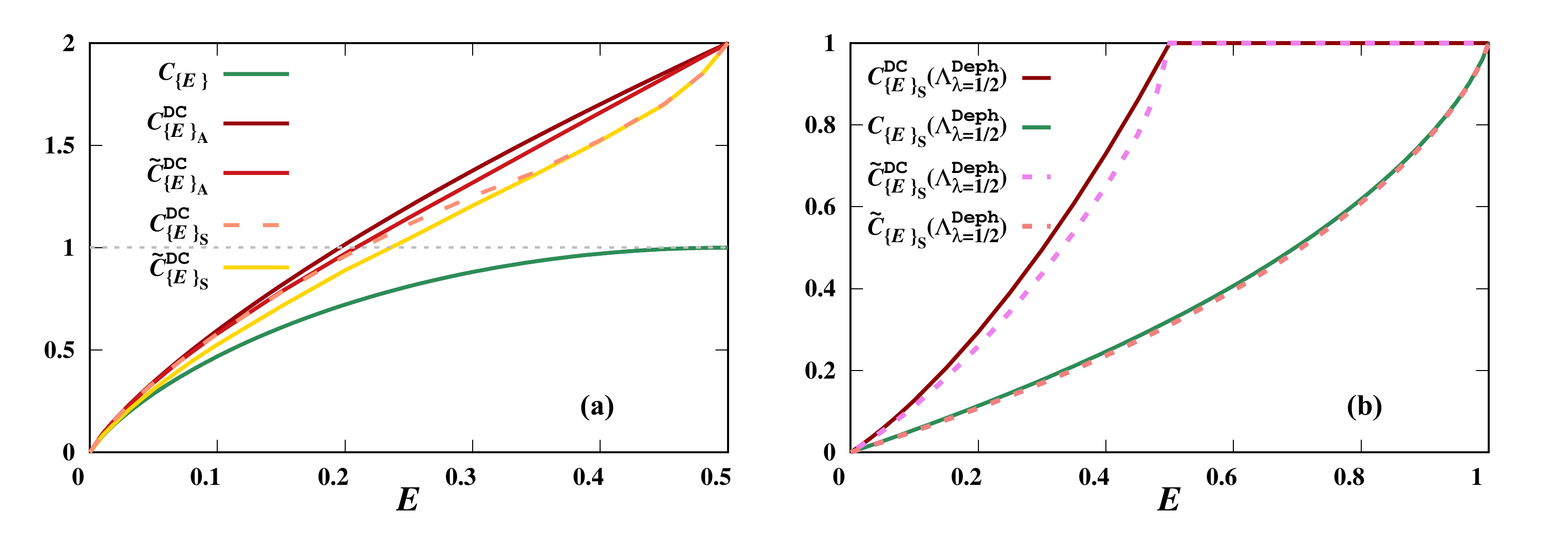}
\caption{(a) comparison of noiseless DC capacities, $C^{\tt DC}_{\{E\}_A}, \tilde C^{\tt DC}_{\{E\}_A}, C^{\tt DC}_{\{E\}_S},\tilde C^{\tt DC}_{\{E\}_S}$ and the unassisted classical capacity $C_{\{E\}}$ (ordinate) in the region of energy bound $E\in[0,1/2]$ (abscissa). (b) DC capacities $C^{\tt DC}_{\{E\}_S}(\Lambda_{\lambda=1/2}^{\tt deph}), \tilde C^{\tt DC}_{\{E\}_S}(\Lambda_{\lambda=1/2}^{\tt deph})$ and entanglement-unassisted classical capacities, $C_{\{E\}_S}(\Lambda_{\lambda=1/2}^{\tt deph}),$ and $\tilde C_{\{E\}_S}(\Lambda_{\lambda=1/2}^{\tt deph})$ (ordinate) of complete dephasing channel against strict energy bound $E$ (abscissa). In the nontrivial energy bound region, i.e., in $E\in(0,1)$, similar to the case of $C^{\tt DC}_{\{E\}_S}(\Lambda_{\lambda=1/2}^{\tt deph})>C_{\{E\}_S}(\Lambda_{\lambda=1/2}^{\tt deph})$ as mentioned in the main text, in the case of equiprobable signal ensemble also, entanglement provides advantage in classical capacity, i.e., $\widetilde C^{\tt DC}_{\{E\}_S}(\Lambda_{\lambda=1/2}^{\tt deph})>\widetilde C_{\{E\}_S}(\Lambda_{\lambda=1/2}^{\tt deph})$. Both axes are dimensionless.} 
\label{fig:compare_dcc_cc}
\end{figure*}
Choosing $q_0=1-E'$ (which is optimal as shown below) leads to the optimal value of the energy of the shared state to be $E'=(1-\sqrt{1-2E})/2$, which yields
\begin{eqnarray}
    \mathcal{C}(E)=2H\left(\frac{1}{2}(1-\sqrt{1-2E})\right).
    \label{eq:qubit_dcc_exp}
\end{eqnarray}
Therefore, for $E>0$, we have $C_{\{E\}_A}^{\tt DC}>C_{\{E\}}$.

\textbf{Proof of optimality of Protocol.~\ref{protocol:ecdc} in $d=2$.} To show that our protocol is optimal in $d=2$ along with the choice of $q_0=1-E'$, which leads to Eq.~\eqref{eq:qubit_dcc_exp}, we resort to numerical optimization technique. Specifically, we consider a set of encoding unitaries \(\{U_i\}_{i=0}^3\), chosen with corresponding probabilities \(\{p_i\}_{i=0}^3\), where \(U_0 = \mathbb{I}\) and, \(U_j = \exp\big(-i \mu_j (\hat{n}_j \cdot \bm{\sigma})\big)\), for \(j \geq 1\). Here, \(\hat{n}_j = (\sin \theta_j \cos \phi_j, \sin \theta_j \sin \phi_j, \cos \theta_j)\) denotes a unit vector on the Bloch sphere, and \(\bm{\sigma} = \{\sigma_x, \sigma_y, \sigma_z\}\) represents the vector of Pauli matrices. Using an algorithm based on improved stochastic ranking strategy (ISRES) \cite{Runarsson2005}, we evaluate
\begin{equation}
\begin{aligned}
    \mathcal C(E) &= \max_{p_0,U_1,U_2,U_3,E'}\chi(\{p_i,\mathcal U_i({\tt Proj}\ket{\psi_{SR}^{E'}})\}_{i=0}^3)\\
    \textrm{subject to} &\sum_{i=0}^3p_i\text{Tr}(\mathcal U_i({\tt Proj}\ket{\psi_{SR}^{E'}})H)\leq E.
    \label{eq:qubit_aec_dcc_numerics}
\end{aligned}
\end{equation}
Here, the maximization is performed over the parameter space $\{0\leq\mu_j,\phi_j\leq2\pi\}_{j=1}^3$, $\{0\leq\theta_j\leq\pi\}_{j=1}^3$, $0\leq p_0\leq 1$ and $0< E'< E$. The  numerical analysis through Eq.~\eqref{eq:qubit_aec_dcc_numerics} matches with the capacity derived in Eq.~\eqref{eq:qubit_dcc_exp} up to $\mathcal O(10^{-5})$. Hence, Protocol.~\ref{protocol:ecdc} is optimal for average energy-constrained DC in $d=2$.

\subsection{Comparing other noiseless constrained DC capacities}
\label{app:onc}
We perform a numerical investigation of DC capacity with strict energy constraint, along with equiprobable messages in both strict and average energy constraints. In the noiseless setting, within the region \( E \in (0, \tfrac{1}{2}) \), we observe (see Fig.~\ref{fig:compare_dcc_cc}(a)) the hierarchy \( C_{\{E\}_A}^{\tt DC} \gtrsim \widetilde{C}_{\{E\}_A}^{\tt DC} \gtrsim C_{\{E\}_S}^{\tt DC} \gtrsim \widetilde{C}_{\{E\}_S}^{\tt DC}>C_{\{E\}} \) where the notational description of DC capacity follows from Eq.~\eqref{eq:dccc} in the main text. Therefore, in any energy constrained or restricted equiprobable signal noiseless scenario, assistance of entanglement provides advantage over classical capacity.

\subsection{DC capacity of complete dephasing channel}
\label{app:cdc}
In the main text, we have shown that $C^{\tt DC}_{\{E\}_s}(\Lambda^{\tt Deph}_{\lambda=\frac12})>C_{\{E\}_s}(\Lambda^{\tt Deph}_{\lambda=\frac12})$ for all values of $E\in(0,1)$. This means the equivalency between entanglement-assisted capacity and unassisted capacity of CQ channel in the energy unconstrained scenario \cite{Shirokov2012} is broken in the strict energy constrained case. In fact, when we restrict our case to the equiprobable signal ensemble, it again follows $\tilde C^{\tt DC}_{\{E\}_s}(\Lambda^{\tt Deph}_{\lambda=\frac12})>\tilde C_{\{E\}_s}(\Lambda^{\tt Deph}_{\lambda=\frac12})=H(E/2)-H(E)/2$ in $E\in(0,1)$ (see Fig.~\ref{fig:compare_dcc_cc}(b)). However, the equivalency is recovered in the case of average energy constrained case, i.e., $C^{\tt DC}_{\{E\}_A}(\Lambda^{\tt Deph}_{\lambda=\frac12})=C_{\{E\}_s}(\Lambda^{\tt Deph}_{\lambda=\frac12})$  (up to numerical accuracy of $\mathcal O(10^{-5})$).

\end{document}